\documentclass[a4paper,11pt]{article}

\usepackage{microtype}

\usepackage[dvipsnames,usenames]{color}
\usepackage[colorlinks=true,urlcolor=Blue,citecolor=Green,linkcolor=BrickRed]{hyperref}
\usepackage[utf8]{inputenc}
\usepackage{algorithm}
\usepackage{algorithmicx}
\usepackage[noend]{algpseudocode}
\usepackage{amsthm}
\usepackage{amssymb}
\usepackage{amsmath}
\usepackage{amsfonts}
\usepackage{graphicx}
\usepackage[noadjust]{cite}
\usepackage{thmtools}
\usepackage{thm-restate}
\usepackage{authblk}
\usepackage{fullpage}

\theoremstyle{plain}
\newtheorem{theorem}{Theorem}
\newtheorem{lemma}{Lemma}
\newtheorem{corollary}{Corollary}
\newtheorem{observation}{Observation}
\newtheorem{proposition}{Proposition}
\newtheorem{definition}{Definition}

\newcommand{\abs}[1]{{\lvert #1 \rvert}}
\newcommand{\inc}{\mathsf{inc}}
\newcommand{\dec}{\mathsf{dec}}
\newcommand{\res}{\mathsf{res}}

\newcommand{\ceil}[1]{\left\lceil #1 \right\rceil}

\algblockdefx{MRepeat}{EndRepeat}{\textbf{repeat}}{}
\algnotext{EndRepeat}

\title{Fast and Longest Rollercoasters}
\date{}
\author[1]{Pawe\l{} Gawrychowski}
\author[2]{Florin Manea}
\author[1]{Rados\l{}aw Serafin}

\affil[1]{Institute of Computer Science, University of Wroc\l{}aw, Poland\\ \tt{gawry@cs.uni.wroc.pl,radserafin@gmail.com}}
\affil[2]{Kiel University, Germany\\ \tt{flm@informatik.uni-kiel.de}}

\begin{document}
\maketitle

\begin{abstract}
For $k\geq 3$, a {\em k-rollercoaster} is a sequence of numbers whose every maximal contiguous subsequence, that is increasing or decreasing, has length at least $k$; $3$-rollercoasters are called simply rollercoasters. Given a sequence of distinct real numbers, we are interested in computing its maximum-length (not necessarily contiguous) subsequence that is a $k$-rollercoaster. Biedl et al. (2018) have shown that each sequence of $n$ distinct real numbers contains a rollercoaster of length at least $\lceil n/2\rceil$ for $n>7$, and that a longest rollercoaster contained in such a sequence can be computed in $O(n\log n)$-time (or faster, in $O(n \log \log n)$ time, when the input sequence is a permutation of $\{1,\ldots,n\}$). They have also shown that every sequence of $n\geqslant (k-1)^2+1$ distinct real numbers contains a $k$-rollercoaster of length at least $\frac{n}{2(k-1)}-\frac{3k}{2}$, and gave an $O(nk\log n)$-time (respectively, $O(n k\log \log n)$-time) algorithm computing a longest $k$-rollercoaster in a sequence of length $n$ (respectively, a permutation of $\{1,\ldots,n\}$).

In this paper, we give an $O(nk^2)$-time algorithm computing the length of a longest $k$-rollercoaster contained in a sequence of $n$ distinct real numbers; hence, for constant $k$, our algorithm computes the length of a longest $k$-rollercoaster in optimal linear time. The algorithm can be easily adapted to output the respective $k$-rollercoaster. In particular, this improves the results of Biedl et al. (2018), by showing that a longest rollercoaster can be computed in optimal linear time. We also present an algorithm computing the length of a longest $k$-rollercoaster
in $O(n \log^2 n)$-time, that is, subquadratic even for large values of $k\leq n$. Again, the rollercoaster can be easily retrieved. 
Finally, we show an $\Omega(n \log k)$ lower bound for the number of comparisons in any comparison-based algorithm computing the length of a longest $k$-rollercoaster.
\end{abstract}

\section{Introduction}
The mathematical study of patterns occurring in sequences of numbers is a rather old and well developed topic in combinatorics and algorithms on sequences. Within this topic, of a particularly high interest is the study of long increasing and decreasing (not necessarily contiguous) subsequences occurring in a sequence. For example, already in 1749, Euler defined the Eulerian polynomials, which are the generating function for the number of descents in permutations. Almost 200 years later, Erd\H{o}s and~Szekeres~\cite{Erdos1935} proved the existence of an increasing or a decreasing subsequence of length at least $a+1$ in a sequence of at least $n=a^2+1$ distinct reals. More precisely, they have shown the following theorem.
\begin{theorem}[Erd\H{o}s and~Szekeres, 1935]
	\label{thr-ES}
	Every sequence of $ab+1$ distinct real numbers contains an increasing subsequence of length at least $a+1$ or a decreasing subsequence of length at least $b+1$.
\end{theorem}
The theorem of Erd\H{o}s--Szekeres is strongly related to, and in fact also follows from, the well-known decomposition of Dilworth (see~\cite{Steele1995}) regarding chains and antichains in a finite partially ordered set. Dilworth's result can be restated in the context of the combinatorics of patterns in sequences of numbers as follows.
\begin{theorem}[Dilworth, 1950]
	\label{thr-Dilworth}
	Any finite sequence $S$ of distinct real numbers can be partitioned into $k$ ascending sequences, where $k$ is the maximum length of a descending sequence in $S$.
\end{theorem}
Recent surveys on the combinatorics of patterns occurring in sequences are~\cite{Linton:2010,Kitaev:2011}.  

The study of patterns in sequences of numbers also has a well developed algorithmic side (see, e.g.,~\cite{Fredman75,Hunt,Crochemore,BESPAMYATNIKH20007}). For instance, finding a longest increasing subsequence (not necessarily contiguous) contained in the input sequence is a basic problem in theoretical computer science, studied already from the 1960s~\cite{eaea,Aldous99,Romik:2015}, with applications in areas such as bioinfomatics and physics (see~\cite{Sun2007} and the references therein).
In particular, in 1975 Fredman~\cite{Fredman75} presented an algorithm (which he attributed to Knuth, now considered folklore) computing the length of a longest
increasing subsequence (LIS) in an array of $n$ numbers in $O(n\log n)$ time, and proved that this is optimal for comparison-based
algorithms. If required, the algorithm can be extended to retrieve such a subsequence.
If the input sequence can be sorted in linear time (in particular, when the input sequence is a permutation of $\{1,\ldots,n\}$)
and we do not require the algorithm to be comparison-based,
the solution given by Fredman can be implemented in $O(n \log\log n)$ time, see~\cite{Crochemore} and the references therein.
Fredman's algorithm is often called \textsc{Patience Sorting}, and has some connections to constructing the so-called Young Tableaux~\cite{Aldous99,Romik:2015}.

We consider a notion that is strongly related to longest increasing subsequences (and longest decreasing subsequences).
A {\em run} in a sequence of numbers is a maximal contiguous subsequence that is either increasing or decreasing.
A {\em $k$-rollercoaster}, where $k\geq 3$, is a sequence of numbers whose every run has length at least $k$; $3$-rollercoasters are
called, for short, rollercoasters. For example, the sequence $(3,6,8,10,9,5,1,2,4,7,11)$ is a $4$-rollercoaster with runs $(3,6,8,10)$,
$(10,9,5,1)$, $(1,2,4,7,11)$. Given a sequence $S[1:n]=(S[1],S[2],\dots,S[n])$ of $n$ distinct numbers, the $k$-rollercoaster problem is to find a maximum-size set of indices $i_1<i_2<\dots <i_m$ such that $(S[{i_1}],S[{i_2}], \dots,S[{i_m}])$ is a $k$-rollercoaster. In other words, this problem asks for a longest $k$-rollercoaster contained in the input sequence $S$. 

There is a simple, but useful, geometrical interpretation of $k$-rollercoasters. The input sequence $S[1:n]$ can be depicted as a set $P$ of points in the plane by translating, for $i$ from $1$ to $n$, the number $S[i]$ to a point $p_i=(i,S[i])$. In this setting, a $k$-rollercoaster in $S$ translates to a polygonal path in the plane, whose vertices are points of $P$, and such that every maximal sub-path, with positive- or negative-sloped edges, has at least $k$ points. The rollercoaster $(3,6,8,10,9,5,1,2,4,7,11)$ is depicted in the left half of Figure~\ref{pic:rol}. Two $4$-rollercoasters occurring in the sequence $(3,6,1,8,7,17,13,10,11,12,9,5,14,4,2,15,16)$ are depicted in the right half of the same figure.

\begin{figure}[htb]
\centering
\includegraphics[width=0.35\textwidth]{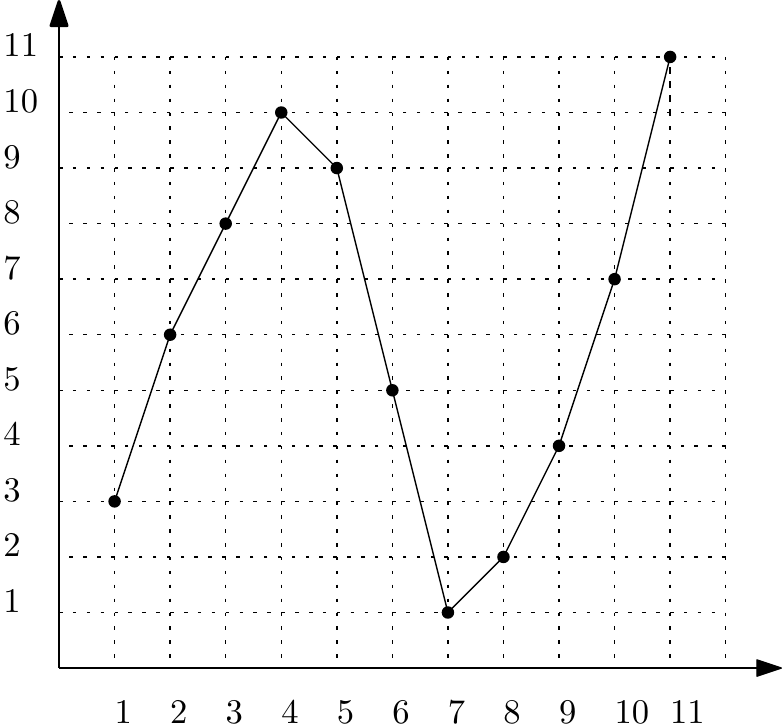} $ $\ \ \ \ 
\includegraphics[width=0.35\textwidth]{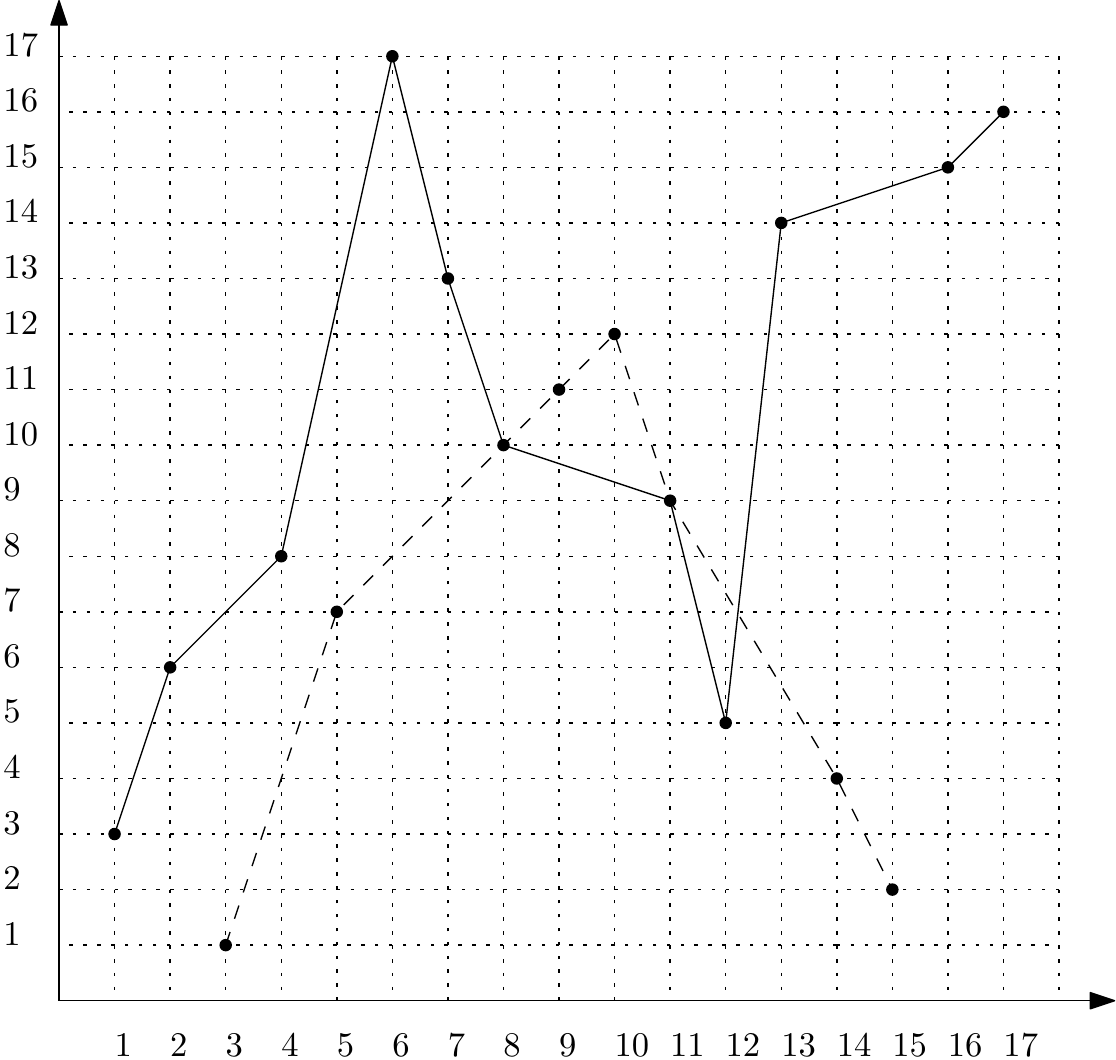}
\caption{Left: a $4$-rollercoaster $(3,6,8,10,9,5,1,2,4,7,11)$ with runs $(3,6,8,10)$, $(10,9,5,1)$, $(1,2,4,7,11)$. Right: two $4$-rollercoasters, represented with a solid and, respectively, a dashed line, in $(3,6,1,8,7,17,13,10,11,12,9,5,14,4,2,15,16)$.}
\label{pic:rol}
\end{figure}

While rollercoasters seem interesting on their own as a combinatorial structure, the original motivation for their study was
a connection to computational geometry and graph drawing, namely to point-set embeddings of {\em caterpillars} (see~\cite{Biedl2017,Biedl2018}
and the references therein). More precisely, constructing a long rollercoaster in a sequence of numbers was used as an intermediate step
towards obtaining a method of drawing a $n$-vertex top-view caterpillar, with L-shaped edges, on a set of $\frac{25}{3}n$
general orthogonal position points in the plane. This is currently the best known bound on the number of points required to draw such a graph.

In~\cite{Biedl2018}, the following results regarding $k$-rollercoasters were shown. First, from a combinatorial point of view, for $k=3$, it was shown that the length of a longest rollercoaster contained in a sequence of $n\geq 7$ distinct numbers is at least $\lceil\frac{n}{2}\rceil$. As far as $k$-rollercoasters are concerned, it was shown that for $k\geqslant 4$ every sequence of $n\geqslant (k-1)^2+1$ distinct numbers contains a $k$-rollercoaster of length at least $\frac{n}{2(k-1)}-\frac{3k}{2}$. From an algorithmic point of view, both previously mentioned results were constructive, leading to an $O(n)$-time (respectively $O(n\log k)$) algorithm computing a long (but not necessarily a longest) rollercoaster (respectively, $k$-rollercoaster) contained in a sequence of $n$ distinct numbers. A longest rollercoaster contained in such a sequence was computed by
an extension of Fredman's algorithm in $O(n\log n)$-time, and if the input sequence is a permutation of $\{1,\ldots,n\}$
(or, more generally, sortable in linear time) in $O(n \log \log n)$ time.
By further generalising this approach, an $O(nk\log n)$-time (respectively, $O(n k\log \log n)$-time) algorithm computing a longest
$k$-rollercoaster in a sequence of $n$ distinct numbers (respectively, a permutation of $\{1,\ldots,n\}$) can be obtained.
Note that, by the theorem of Erd\"os and Szekeres, a sequence of $n$ distinct numbers always contains a $\lfloor\sqrt{n}\rfloor$-rollercoaster,
and the aforementioned algorithm computes a longest such rollercoaster in $O(n^{1.5}\log n)$ time.

\subparagraph{Our contributions.}  We consider the problem of computing a longest $k$-rollercoaster in an input sequence $S[1:n]$ and provide
three results.

Firstly, we design a comparison-based algorithm computing the length of a longest $k$-rollercoaster in a sequence of $n$ distinct numbers
in $O(nk^2)$ time. Thus, we obtain an optimal linear-time algorithm for constant values of $k$, in particular for $k=3$. This
significantly improves the results of~\cite{Biedl2018} and shows that, even though longest rollercoasters are related to longest increasing
subsequences, the rich combinatorial structure of the former makes them provably easier to find.
The starting point of our algorithm is the following natural dynamic programming formulation. For each $2\leq i\leq k$, and for
each element $S[j]$, we compute a longest (not necessarily contiguous) subsequence of $S$ ending with $S[j]$ and with every run of
length at least $k$, except for the last run, which has only $i$ elements if $i < k$ and at least $k$ elements if $i = k$.
Now the difficulty is to find the predecessor $S[j']$ of $S[j]$ in such a subsequence in time proportional to $k$, in particular avoiding
any kind of binary search. We greedily decompose the input sequence into blocks with a certain property related to Dilworth's
theorem and prove, by a careful case analysis, that $j'$ must belong to the previous few such blocks. This, together with
the special structure of the blocks and appropriate data structures, allows us to find $j'$ in $O(k)$ amortised time.

Secondly, we focus on the case of large $k$. Given that both the previous and the new algorithm have at least linear dependency on $k$,
it might seem plausible that this is inherent to the problem, for example that for $k\geq \lfloor\sqrt{n}\rfloor$ the running time
of any algorithm needs to be $\Omega(n^{1.5})$. We show that this is not the case by designing
a subquadratic algorithm that computes a longest $k$-rollercoaster in a sequence of $n$ distinct numbers in $O(n\log^{2}n)$ time.
To obtain this result, we exploit the fact that if an increasing (respectively, decreasing) run in a longest $k$-rollercoaster extends from $S[i]$ to $S[j]$, 
then that run should be LIS (respectively, longest decreasing sequence, LDS for short) in $S[i:j]$. If one arranges the length of
LIS (respectively, LDS) in $S[i:j]$ in an $n\times n$ matrix then the matrix has the anti-Monge property. It is known that
all row maxima of an anti-Monge matrix can be found in $O(n)$ time~\cite{SMAWK}, that is, in sublinear time w.r.t. the size of the matrix (given an oracle
access to the elements of the matrix). Such properties have been successfully exploited to speed up certain dynamic programming algorithms. 
We also follow this route, and construct a longest $k$-rollercoaster using dynamic programming, essentially by gluing together
LISs and LDSs of consecutive contiguous subsequences of $S$.

Thirdly, we show that any comparison-based algorithm computing a longest $k$-rollercoaster needs $\Omega(n\log k)$ comparisons.
Our reasoning is similar to the one used by Fredman to show that any comparison-based algorithm computing a LIS
needs $\Omega(n\log n)$ comparisons. We leave as an open problem to close the gap between the lower and upper bounds shown here.

The paper is organised as follows. After a series of preliminaries, we describe the~$O(nk^2)$-time algorithm for computing
the length of a longest $k$-rollercoaster, followed by the $O(n \log^2 n)$-time algorithm. Then we show how the respective longest $k$-rollercoasters can be effectively constructed. We conclude with the lower bound for
the number of comparisons needed to compute the length of a longest $k$-rollercoaster in a sequence of length~$n$. %The proofs omitted here for space reasons can be found in \cite{roller_Arxiv}.

\section{Preliminaries}
\label{prel}
We consider sequences of distinct real numbers and work in the comparison-based model.
If $S$ is a sequence of $n$ numbers, then $|S|=n$ is the length of the sequence, and $S[i]$ denotes its $i^\text{th}$ element.
A {\em subsequence} of $S$ is a sequence $(S[i_1],S[i_2],\ldots, S[i_m])$, defined by specifying the indices $1\leq i_1<i_2<\ldots<i_m\leq n$.
For $1\leq i\leq j\leq n$, $S[i:j]$ denotes the {\em contiguous subsequence} $(S[i], S[i+1],\ldots, S[j])$; in particular, $S[1:n]$ denotes the entire $S$.
Note that unless explicitly stated, a subsequence is not necessarily contiguous. An increasing subsequence (respectively, decreasing subsequence)
of $S$ is a subsequence $(S[i_1],S[i_2],\ldots, S[i_m])$ such that $S[i_j] < S[i_{j+1}]$, for all $1\leq j\leq m-1$ (respectively, $S[i_j] > S[i_{j+1}]$,
for all $1\leq j\leq m-1$). A longest increasing (respectively, decreasing) sequence, for short LIS (respectively, LDS), is an increasing
(respectively, decreasing) sequence with the largest possible length. Fredman gave an $O(n\log n)$-time algorithm for computing the
length of LIS, denoted $\res$ in Algorithm~\ref{alg:lis}. A byproduct of this algorithm is a partition of $S[1:n]$ into $\res$ non-increasing
subsequences that can be obtained by creating, for every $1\leq j \leq \res$, a list of elements that has been stored in $R[j]$.

\begin{algorithm}
\caption{Finding the length of LIS of $S$}
\label{alg:lis}
\begin{algorithmic}[1]
\State $R[0] \gets 0$
\State $\res \gets 0$
\For{$i \gets 1 $ \textbf{to} $ n$}
\State $ k \gets \max \{ j : R[j] < S[i]\} $ \Comment binary search over $R[0] < R[1] < R[2] < \ldots $
\State $ R[k+1] \gets S[i] $
\State $ \res \gets \max \{ \res , k+1 \}$
\EndFor
\State \textbf{return} $\res$
\end{algorithmic}
\end{algorithm}

A {\em run} in a sequence of numbers is a maximal contiguous subsequence that is increasing or decreasing. A {\em $k$-rollercoaster}
is a sequence of numbers such that every run has length at least $k$; $3$-rollercoasters are called, for short, rollercoasters.
Given a sequence $S[1:n]$ we are interested in finding its longest subsequence that is a $k$-rollercoaster. To make the exposition
easier to follow, we focus first on finding the length of such a subsequence. Recovering the subsequence itself is, in all our algorithms,
rather straightforward, and explained in Section \ref{sec:recover}.

\section{Computing a Longest $k$-Rollercoaster in $O(nk^2)$-Time}
In this section we show how to find a longest $k$-rollercoaster of $S[1:n]$ in $O(nk^{2})$ time.

We begin our algorithm with a preprocessing phase. An alternating $k$-decomposition of $S[1:n]$ is a partition
of $S[1:n]$ into contiguous subsequences (called parts) $S_1,S_2,\ldots,S_m$ such that the length of LIS in
the odd parts ($S_1$, $S_3$, $S_5$, and so on) is $k$ while the length of LDS in the even parts is $k$, possibly
smaller for the very last part, and additionally by removing the last
element of any odd (even) part we obtain a sequence with LIS (LDS) of length less than $k$.
In other words, for $\ell \geq 1$, $S_\ell$ is either the shortest contiguous subsequence of $S$ that follows directly after $S_1\cdots S_{\ell-1}$
and has for $\ell$ odd (even) a LIS (respectively, LDS) of length $k$, if such a subsequence exists, or the whole
remaining part of $S$ otherwise.
For example, an alternating 3-decomposition of $S = (1, 4, 2, 5, 8, 7, 6, 3)$ is $(1,4,2,5), (8,7,6), (3)$.

\begin{lemma}
\label{lem:decompose}
An alternating $k$-decomposition of $S[1:n]$ can be found in $O(n\log k)$ time.
\end{lemma}
\begin{proof}
By terminating Algorithm~\ref{alg:lis} as soon as $\res=k$ we can find the shortest prefix of $S$ with LIS equal to $k$
in $O(d\log k)$ time, where $d$ is the length of the prefix. Then we find the shortest prefix of the remaining suffix
of $S$ with LDS equal to $k$, and repeat. Overall, this takes $O(n\log k)$ time because all parts are disjoint.
\end{proof}

\begin{proposition}
\label{prop:alternations}
Let $A$ be a $k$-rollercoaster in $S$. Any part $S_\ell$ contains elements
of at most four consecutive runs of $A$.
\end{proposition}
\begin{proof}
By contradiction. Let $S'_{\ell}$ be $S_{\ell}$ without the last element.
If $S_{\ell}$ contains elements of five consecutive runs of $A$ then
$S'_{\ell}$ contains elements of four consecutive runs of $A$,
and hence all elements of two such consecutive runs.
Thus, if $S_{\ell}$ is an odd (even) part then $S'_{\ell}$ contains LIS (LDS)
of length $k$, which contradicts the definition of an alternating
$k$-decomposition.
\end{proof}

By Dilworth's theorem, a part with LIS of length $k$ can be decomposed into $k$ decreasing subsequences, and such
a decomposition can be obtained as a byproduct of Algorithm~\ref{alg:lis}. Thus, we can decompose each part into
up to $k$ monotone (increasing or decreasing, depending on whether the part is odd or even) subsequences.
These subsequences can be then merged to obtain a sorted list $P_{\ell}$ of all elements in the corresponding part
$S_{\ell}$ in $O(n\log k)$ overall time, for example by first merging pairs of subsequences, then quadruples, and so on.

Before moving on to the description of our algorithm, we need a combinatorial lemma that relates an alternating
$k$-decomposition to a longest rollercoaster.

\begin{restatable}{lem}{lemsplit}
\label{lem:split}
Suppose that $x=S[j]$ is a non-first element occurring in an increasing run of a longest $k$-rollercoaster, and $y=S[j']$ is
its predecessor in the same run, and consider an alternating $k$-decomposition of $S[1:n]$. Then either $x$ and $y$ are
in the same part $S_i$, or $y$ is in one of the parts $S_{i-4}, S_{i-3}, S_{i-2}, S_{i-1}$. 
\end{restatable}
\begin{proof}
By contradiction. Suppose that there are at least four parts between $x$ and $y$, i.e., $x$ is in $S_i$ and $y$ is in some $S_k$ with $k<i-4$. Let $r$ denote the run in the $k$-rollercoaster that contains $x$ and $y$, let $d$ be the length of $r$, and let $\ell$ be such that $r[\ell] = y$ and $r[\ell+1] = x$. We assume that $r$ is an increasing run (see Figure \ref{pic:rol2}); the case when $r$ is decreasing can be treated in the same way. 

\begin{figure}[htb]
\centering
\includegraphics[width=0.5\textwidth]{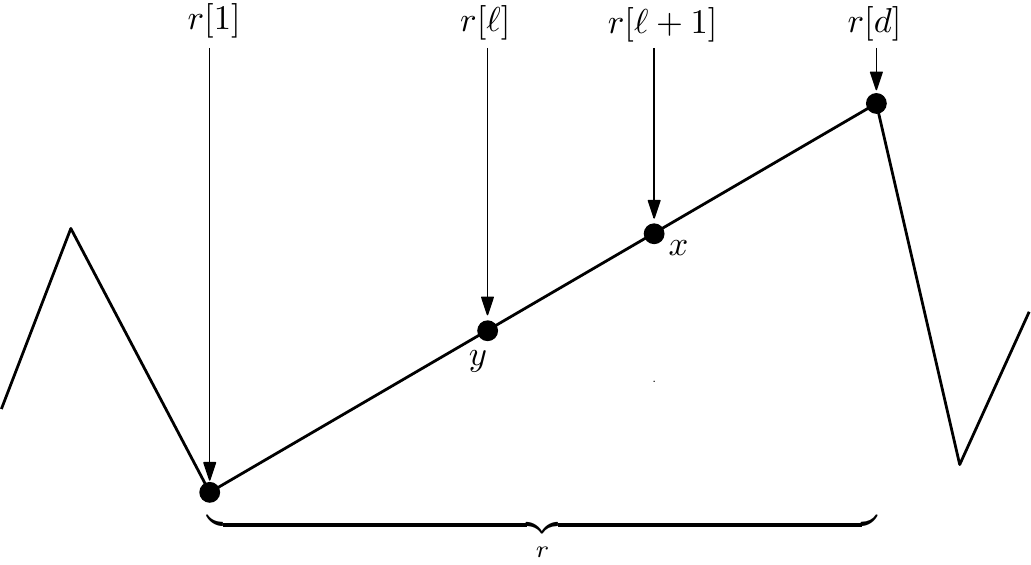}
\caption{The increasing run $r$ from Lemma \ref{lem:split}, with the points $x$ and $y$ highlighted.}
\label{pic:rol2}
\end{figure}
Consider the following four cases:
\begin{enumerate}
\item $\ell\leq k-1$ (i.e., there are at most $k-2$ elements in $r$ before $y$) and $ k-2\geq d -\ell -1$ (there are at most $k-2$ elements in $r$ after $x$).
\item $\ell\leq k-1$  and $k-1\leq d-\ell -1$ (there at least $k-1$ elements in $r$ after $x$).
\item $\ell \geq k$ (there are at least $k-1$ elements in $r$ before $y$) and $ k-2\geq d -\ell -1$.
\item $\ell \geq k$ and  $k-1\leq d-\ell -1$.
\end{enumerate}

Recall that there are at least four whole parts between $x$ and $y$. Therefore, in particular there are three consecutive parts
$S_{i'}, S_{i'+1}$, and $S_{i'+2}$ such that the first has LIS of length $k$, the second has LDS of length $k$, and the third
has LIS of length $k$.

In the first case, we replace $r[2:d-1]$ with LIS of $S_{i'}$, the LDS of $S_{i'+1}$, and LIS of $S_{i'+2}$.
It is straightforward to verify that we obtain a valid $k$-rollercoaster, and because we remove at most $2k-4$ elements and
add at least $3k$, this creates a longer $k$-rollercoaster, which is a contradiction.
In the second case, we replace $r[2:\ell]$ with LIS of $S_{i'}$ and LDS of $S_{i'+1}$. Again, it is straightforward
to verify that we obtain a valid longer $k$-rollercoaster, because we remove at most $k-2$ elements and add
at least $2k$.
Similarly, in the third case, we replace $r[\ell+1 : d-1]$ with LDS of $S_{i'+1}$ and LIS of $S_{i'+2}$ to obtain a longer $k$-rollercoaster.
Finally, in the fourth case we simply insert LDS of $S_{i'+1}$ between $x$ and $y$ to obtain a longer $k$-rollercoaster.
\end{proof}

After the initial preprocessing phase we apply dynamic programming.
For $1\leq i\leq k$, we say that a subsequence of $S$ (not necessarily contiguous) is a $(k,i)_{+}$-rollercoaster
if it ends with an increasing run of length exactly $i$ when $i<k$ and at least $k$ when $i=k$, while every other
run is of length at least $k$. Additionally, we consider $k$-rollercoaster ending with a decreasing run as $(k, 1)_{+}$-rollercoaster. We want to construct, for every $1\leq i \leq k$ and $1\leq j\leq n$, a longest $(k,i)_{+}$-rollercoaster
ending with $S[j]$. To this end we calculate $M_{+}[j,i]$, the position in $S$ of the predecessor of $S[j]$ in such
a $(k,i)_{+}$-rollercoaster, and $L_{+}[j,i]$, the length of the respective $(k,i)_{+}$-rollercoaster.
A $(k,i)_{-}$-rollercoaster is defined similarly, except that the last run should be decreasing, and we also calculate
the values $M_{-}[j,i]$ and $L_{-}[j,i]$, defined similarly to the above and corresponding to such a $(k,i)_{-}$-rollercoaster.
We only describe in detail how to compute $M_{+}[j,i]$ and $L_{+}[j,i]$, as $M_{-}[j,i]$ and $L_{-}[j,i]$ are computed analogously.
The computation proceeds from left to right, that is, we iterate over the parts $S_1, S_{2}, \ldots$ and
compute, for every element $S[j]$ of the current part $S_{\ell}$, the values of $M_{+}[j,i]$ and $L_{+}[j,i]$ for every
$1\leq i\leq k$. See Algorithm~\ref{alg:nk2} for a high-level overview of the algorithm.

\begin{algorithm}[h]
\caption{Computing the length of a longest $k$-rollercoaster}
\label{alg:nk2}
\begin{algorithmic}[1]
\State Find an alternating $k$-decomposition $S_{1},\ldots,S_{m}$ of $S$.
\For{$1 \leq \ell \leq m$}
\State Merge the $k$ monotone subsequences constituting $S_{\ell}$ to obtain a single sorted list $P_{\ell}$.
\EndFor
\For {$1\leq \ell \leq m$}
\Statex \quad\ $\triangleright$ For each $S[j]$ in $S_\ell$ and $1\leq i\leq k$, we compute the following:
\Statex \quad\ \ \ \ \ \ $M_+[j,i]$: position in $S$ of the predecessor of $S[j]$ in its $(k,i)_{+}$-rollercoaster
\Statex \quad\ \ \ \ \ \ $L_+[j,i]$: length of the respective $(k,i)_{+}$-rollercoaster
\For{$2\leq i\leq k$}
\For{$1\leq d \leq 4$ and each $S[j]\in S_{\ell}$ in the order of their occurrences in $P_{\ell}$}
\State Find $M^d_+[j,i]$ and $L^d_+[j,i]$.
\EndFor
\For{each $S[j]\in S_{\ell}$}
\State $L_{+}[j,i] \gets \max \{  L^{d}_{+}[j,i] : 1\leq d \leq 4\}$
\State Set $M_{+}[j,i]$ so that it corresponds to $L_{+}[j,i]$.
\EndFor
\EndFor
\State Compute, for each $S[j] \in S_{\ell}$, $L_{-}[j,i]$ and $M_{-}[j,i]$ with a similar approach.
\MRepeat \ 4 times
\For{each $S[j]\in S_{\ell}$}
\State $L_{+}[j,1] \gets \max \{L_{-}[j,k] ,1\}$
\State $M_{+}[j,1] \gets M_{-}[j,k]$ if $L_{-}[j,k] > 0$ and $0$ otherwise
\EndFor
\For {$2\leq i\leq k$} 
\For{each $S[j]\in S_{\ell}$ in the order of their occurrences in $P_{\ell}$}
\State Find $M'_+[j,i]$, $L'_+[j,i]$ using decomposition of $S_{\ell}$ into $k$ monotone sequences.
\EndFor
\For{each $S[j]\in S_{\ell}$}
\State $L_{+}[j,i] \gets \max\{  L_{+}[j,i], L'_{+}[j,i] \}$
\State Update $M_{+}[j,i]$ so that it corresponds to $L_{+}[j,i]$.
\EndFor
\EndFor
\State Update, for each $S[j] \in S_{\ell}$, $L_{-}[j,i]$ and $M_{-}[j,i]$ with a similar approach.
\EndRepeat
\EndFor
\State \Return $\max\{ \max\{ L_-[j,k],L_+[j,k] \} : 1 \leq j \leq n \}$
\end{algorithmic}
\end{algorithm}

When we begin computing the arrays $M_+[\cdot,i]$, $L_+[\cdot,i]$, $M_-[\cdot,i]$ and $L_-[\cdot,i]$, for $2\leq i\leq k$,
corresponding to all $S[j] \in S_{\ell}$, we have already computed
$M_+[j',1]$, $M_+[j',2], \ldots,$ $M_+[j',k]$ and $L_+[j',1]$, $L_+[j',2], \ldots,$ $L_+[j',k]$, as well as
$M_-[j',1]$, $M_-[j',2], \ldots,$ $M_-[j',k]$ and $L_-[j',1]$, $L_-[j',2], \ldots,$ $L_-[j',k]$, for every $S[j']\in S_{\ell'}$
such that $\ell'<\ell$.

We start with computing the values $M_+[\cdot,i]$, $L_+[\cdot,i]$, $M_-[\cdot,i]$ and $L_-[\cdot,i]$, for $2\leq i\leq k$,
assuming that the predecessor $S[j']$ of $S[j]$ in its corresponding rollercoaster belongs to $S_{\ell-d}$, for some $1\leq d \leq 4$.
In such case the longest rollercoaster ending at $S[j']$ has been already correctly determined and the
computation is quite straightforward.
If $S[j']$ also belongs to $S_{\ell}$, we must be more careful to guarantee that the longest rollercoaster ending
at $S[j']$ is already known.
We proceed in iterations. In the $t^\text{th}$ iteration, we guarantee to compute the values such that
at most $t$ runs of the corresponding rollercoaster contain elements from $S_{\ell}$.
By Proposition~\ref{prop:alternations}, four iterations are enough. 
In a single iteration, we start with computing the initial values $M_+[\cdot,1]$, $L_+[\cdot,1]$, $M_-[\cdot,1]$ and $L_-[\cdot,1]$
corresponding to $S[j]$ being the first element of its run.
These values can be simply copied from the already known 
$M_-[\cdot,k]$, $L_-[\cdot,k]$, $M_+[\cdot,k]$ and $L_+[\cdot,k]$ corresponding to $S[j]$ being the last element
of a rollercoaster with less than $t$ runs containing elements from $S_{\ell}$ (or set to 1 corresponding to $S[j]$ being the only
element in the rollercoaster).
This is correct because a $(k,1)_{+}$-rollercoaster is actually either a $(k,k)_{-}$-rollercoaster or a sequence consisting of a single element.
Then, we calculate the values $M_+[\cdot,i]$, $L_+[\cdot,i]$, $M_-[\cdot,i]$ and $L_-[\cdot,i]$, for $2\leq i<k$,
such that the predecessor $S[j']\in S_{\ell}$ belongs to the same run as $S[j]$.
By performing the calculation for $i=2,3\ldots,k-1$ in this order we guarantee that the longest rollercoaster
ending at the predecessor $S[j']\in S_{\ell}$ is already known for all $S[j]\in S_{\ell}$, but the computation is
still not completely trivial and requires a different approach depending on whether $S_{\ell}$ was decomposed
into at most $k$ increasing, respectively decreasing, subsequences.
Finally, we extend this to $i=k$.

\subparagraph{$M_+[j,i]$ belongs to $S_{\ell-d}$ for some $1\leq d \leq 4$. } We process $S_{\ell-d}$ to identify
some candidates, denoted $M^d_+[j,i]$ and $L^d_+[j,i]$, for $M_+[j,i]$ and $L_+[j,i]$, respectively, for every $S[j] \in S_{\ell}$.
The idea is to compute these candidates in the order in which the elements $S[j]$ occur on the sorted list $P_{\ell}$.
So, let us consider $P_{\ell}$ and $P_{\ell-d}$. For each element $S[j]$ in the current part we want to identify
a longest $(k,i-1)_{+}$-rollercoaster ending in $S_{\ell-d}$ with an element less than $S[j]$. Thus, as $P_{\ell-d}$ is increasing,
for every element of the current part we need to consider all elements in a prefix of $P_{\ell-d}$. 
Also, if $S[j']$ is to the right of $S[j]$ in $P_\ell$, that is, $S[j'] \geq S[j]$, then the prefix of $P_{\ell-d}$
that we need to consider to compute $M^d_+[j',i]$ is at least as long as the prefix that we need to consider to
compute $M^d_+[j,i]$. Therefore, we can use two pointers to sweep through $P_\ell$ and $P_{\ell-d}$ from left to right, and
obtain the information needed to compute $M^d_+[j,i]$ and $L^d_+[j,i]$, for every $S[j]\in S_\ell$. At the beginning the
pointers point to the first element of $P_\ell$ and $P_{\ell-d}$, respectively. Say that the current element in $P_{\ell}$
and $P_{\ell-d}$ is $S[j]$ and $S[h]$, respectively (we update indices $j$ and $h$ along with the pointers).
We keep moving forward the pointer corresponding to $S[h]$ until we find an element $S[h]>S[j]$. Then we set $M^d_+[j,i]=h'$
and $L^d_+[j,i]=L_+[h',i]+1$, where $S[h']$ is an element occurring earlier than $S[h]$ in $P_{\ell-d}$ with the largest value of
$L_{+}[h',i-1]$. The element $S[h']$ is maintained as we move from left to right in $P_{\ell-d}$. Then we proceed to the
next element in $P_{\ell}$.
Overall, computing candidates $M^d_+[j,i]$ and $L^d_+[j,i]$, for every $S[j]\in S_\ell$, takes $O(|S_{\ell-d}|+|S_{\ell}|)$ time.

\subparagraph{$M_+[j,i]$ belongs to $S_{\ell}$ decomposed into $k$ increasing subsequences. }
Recall that we have already computed $M_+[j',i']$ and $L_+[j',i']$ for every $i'<i$ and $S[j']\in S_\ell$,
and the goal is to identify candidates, denoted $M'_+[j,i]$ and $L'_+[j,i]$, for $M_+[j,i]$ and $L_+[j,i]$, respectively, for every
$S[j]\in S_\ell$. Consider the decomposition of $S_{\ell}$ into $k$ increasing subsequences $I_{1},I_{2},\ldots,I_{k}$.
The elements of every sequence are increasing w.r.t. their values and w.r.t their positions in $S$. 
Consider an element $S[j] \in I_{a}$ and $1\leq b\leq k$ (possibly $a=b$). The elements of $I_{b}$ that can be
the predecessor of $S[j]$ in a $(k,i)_{+}$-rollercoaster (that is, possible candidates for $M_{+}[j,i]$) are both less w.r.t. value and w.r.t. position in $S$.
Thus, these elements form a prefix of $I_{b}$, and for every $S[j] \in I_{a}$ and $1\leq b\leq k$ we want to
maximise $L_{+}[h',i]$ over all $S[h']$ in such a prefix. As in the previous case, we can use two pointers to
sweep through $I_{a}$ and $I_{b}$ and compute, for every $S[j] \in I_{a}$, the element $S[h']\in I_{b}$ that
could precede $S[j]$ in a $(k,i)_{+}$-rollercoaster with the largest value of $L_{+}[h',i-1]$. Finally, we set
$M'_+[j,i]$ and $L'_+[j,i]$ to correspond to the largest such value among all $1\leq b\leq k$.
Overall, computing the candidates $M'_+[j,i]$ and $L'_+[j,i]$, for every $S[j]\in S_\ell$, takes $O(k|S_{\ell}|)$ time.

\subparagraph{$M_+[j,i]$ belongs to $S_{\ell}$ decomposed into $k$ decreasing subsequences.}
This is the most complicated case. Recall that the decomposition into $k$ decreasing subsequences $D_{1},D_{2},\ldots,D_{k}$
was obtained with Algorithm~\ref{alg:lis}. In more detail, $D_{a}$ consists of elements assigned to $R[a]$ throughout
the execution of the algorithm. Thus, if $S[j] \in D_{a}$ then the predecessor of $S[j]$ in a sought longest
$(k,i)_{+}$-rollercoaster, denoted $S[j']$, must belong to $D_{b}$ for some $1\leq b < a$. Indeed,
Algorithm~\ref{alg:lis} first processes $S[j']$ and then $S[j]$, so if $S[j'] \in D_{b}$ then $R[b] \leq S[j']$
when processing $S[j]$ and consequently $S[j'] < S[j]$ implies that $S[j]$ is assigned to $R[a]$ with $a>b$.
So, we first compute the candidates $M'_+[j,i]$ and $L'_+[j,i]$ for every $S[j] \in D_{1}$, then for every $S[j] \in D_{2}$,
and so on.
That is, consider a decreasing subsequence $D_a$ and suppose that we have already computed the desired result for all elements
in $D_{1}, D_{2}, \ldots, D_{a-1}$. Note that at this point we have already computed, for every $S[j] \in D_{1} \cup \ldots \cup D_{a-1}$,
the values of $M^d_+[j,i]$ and $L^d_+[j,i]$, for $1\leq d\leq 4$, as well as the values $M'_+[j,i]$ and $L'_+[j,i]$
corresponding to the current iteration. Thus, we
are already able to set $M_+[j,i]$ and $L_+[j,i]$ by choosing the option that maximises the length of the corresponding $(k,i)_{+}$-rollercoaster,
which is important when extending this case to $i=k$.

Consider an element $S[j] \in D_{a}$ and $1\leq b < a$. 
The elements of $D_{b}$ that can be the predecessor of $S[j]$ in a $(k,i)_{+}$-rollercoaster (that is, possible candidates for $M_{+}[j,i]$)
are both less w.r.t. value and w.r.t. position in $S$, similarly as in the previous case.
The difference is that now these elements form contiguous subsequence $X$ of $D_{b}$ that is not necessarily a prefix.
The first element of $X$ can be found by searching for the first element with sufficiently small value, while
its last element can be found by searching the last element with sufficiently small position (note that $X$ might be empty).
Let $S[j']$ be the next element after $S[j]$ in $D_{a}$, and $Y$ be its corresponding contiguous subsequence of $D_{b}$
consisting of possible predecessors in a $(k,i)_{+}$-rollercoaster. Clearly, $S[j] > S[j']$ while $j < j'$.
Thus, the first element of $Y$ is either the same as the first element of $X$ or occurs after the first element of $X$
in $D_{a}$, while the last element of $Y$ is either the same as the last element of $X$ or occurs after the last element
of $X$ in $D_{a}$ (assuming that both $X$ and $Y$ are non-empty). Thus, we sweep through $D_{a}$ while maintaining the current contiguous subsequence $X$ of $D_{b}$
corresponding to the possible predecessors of the current $S[j]\in D_{a}$. This requires the following tool.
\begin{lemma}[\cite{Tarjan}]
\label{lem:deque}
There is a data structure that maintains a list of elements under the following operations: pop an element from the front, push
an element in the back, and return the maximum element in the current list, each in $O(1)$ time.
\end{lemma}
When processing the current element $S[j] \in D_{a}$ we maintain the first element $S[f]\in D_{b}$ such that $S[f] < S[j]$
and the last element $S[\ell]\in D_{b}$ such that $\ell < j$. Then $X$ consists of all elements between $S[f]$
and $S[\ell]$ in $D_{b}$ (inclusive), and is maintained in a structure from Lemma~\ref{lem:deque} storing
the lengths of their corresponding $(k,i)_{+}$-rollercoaster, that is, the already
known value of $L_+[\cdot,i-1]$. This allows us to extract the element $S[j'] \in X$ with the largest value of
$L_+[j',i-1]$, and set $M'_{+}[j,i]=j'$ and $L'_{+}[j,i]=L_{+}[j',i-1]+1$ in constant time, while updating $f$ and $\ell$
takes amortised constant time. 
Overall, computing the candidates $M'_+[j,i]$ and $L'_+[j,i]$, for every $S[j]\in S_\ell$, takes $O(k|S_{\ell}|)$ time.

\subparagraph{Case $i=k.$ }
To compute $M_+[j,k]$ and $L_+[j,k]$, we first use exactly the same approach as before
for $i=k$, so consider the values of $M_+[\cdot,k-1]$ and $L_+[\cdot,k-1]$. But this only allows us to compute
the length of a longest $(k,k)_{+}$-rollercoaster with the last run of length exactly $k$. To extend this to
arbitrary $(k,k)_{+}$-rollercoasters with the last run of length greater than $k$ we additionally run the same
algorithm but instead of looking at $M_+[\cdot,k-1]$ and $L_+[\cdot,k-1]$ we use $M_+[\cdot,k]$ and $L_+[\cdot,k]$,
including the values already computed in this extra step in the third case.
The reason why this works is that, due to the order in which we consider the elements of $S_\ell$, at the moment when
we compute the length of a longest $(k,k)_{+}$-rollercoaster ending with $S[j]$, and which may have more than $k$ elements
in the final run, we have already computed the length of a longest $(k,k)_{+}$-rollercoaster ending with any element $S[j']$ which
may be a predecessor of $S[j]$ on the respective $(k,k)_{+}$-rollercoaster.
 
\subparagraph{Conclusion.} With these final remarks, our algorithm is completely described. It only remains to find the element $S[j]$ for which $\max\{L_-[j,k],L_+[j,k]\}$ is maximum.
The correctness follows from the comments made throughout its description. To compute the complexity, it is enough
to note that each part $S_{\ell}$ of the partition of $S$ is processed in $O(k|S_\ell|)$ time, for each $2\leq i\leq k$. Adding
this up, we get that the total complexity of our algorithm is $O(nk^2)$.

\begin{theorem}\label{nk-squared}
For every sequence $S[1:n]$ and $k\geq 3$, the length of a longest $k$-rollercoaster in $S$ can be found in $O(nk^2)$-time.
\end{theorem}

\section{Computing a Longest $k$-Rollercoaster in 
$O(n \log^2{n})$-time}

Before we describe our algorithm, we introduce two preliminary procedures. Firstly, we introduce the definition of an anti-Monge matrix
and the algorithm for finding the maximum in every column of such a matrix. Secondly, we describe the algorithm for finding LIS
in contiguous subsequences of the input sequence. Finally, we describe the algorithm computing a longest $k$-rollercoaster in this 
sequence, using the previously developed tools as black~boxes.

\subparagraph{Monge matrices.}
\label{ss:monge}
Let $A$ be an $n\times n$ matrix, and $A[i, j]$ denote its element in the $i^\text{th}$ row from the top and the $j^\text{th}$ column from the left.
$A$ is {\em Monge} (respectively, {\em anti-Monge}) if, for every $1 \leq i < j \leq n$ and $1 \leq k < \ell \leq n$, the {\em Monge equality} holds,
namely $A[i, k] + A[j, \ell] \leq A[i, \ell] + A[j, k]$ (respectively, $A[i, k] + A[j, \ell] \geq A[i, \ell] + A[j, k]$). 
An $n \times n$ {\em falling staircase anti-Monge matrix} is a matrix with blanks such that for every blank all elements below and to the left are blanks,
and the anti-Monge inequality holds whenever the four concerned elements are non-blank.
Similarly, an $n \times n$ {\em reverse falling staircase anti-Monge matrix} is a matrix with blanks such that for every blank all elements above and to the right are blanks, and the anti-Monge inequality holds whenever the four concerned elements are non-blank.
Finally, an $n\times n$ matrix $A$ is {\em totally monotone} if, for every $1 \leq i < j \leq n$ and $1 \leq k < \ell \leq n$,
$A[i,k] \leq A[i,\ell]$ implies $A[j,k] \leq A[j,\ell]$.
\begin{figure}[htb]
\centering
\begin{minipage}{0.32\textwidth}
\centering
\begin{tabular}{|c|c|c|c|c|}
  \hline 
  0 & 1 & 2 & 2 & 2\\
  \hline
  -1 & 0 & 1 & 1 & 2\\
  \hline
  -2 & -1 & 0 & 1 & 2\\
  \hline
  -3 & -2 & -1 & 0 & 1\\
  \hline
  -4 & -3 & -2 & -1 & 0\\
  \hline
\end{tabular}
\end{minipage}
\hfill
%%%%%%%%%%%%%%%%%%%%%%%%%%%%%%%%%%%%%%%%%%%%
\begin{minipage}{0.32\textwidth}
\centering
\begin{tabular}{|c|c|c|c|c|}
  \hline 
    &   &   &   &  \\
  \hline
   4 &  &  &  & \\
  \hline
   3 & 2 & &  & \\
  \hline
   2 &  1 &  & & \\
  \hline
   4 & 4 & 2 &  6 & \\
  \hline
\end{tabular}
\end{minipage}
\hfill
%%%%%%%%%%%%%%%%%%%%%%%%%%%%%%%%%%%%%%%%%%%%
\begin{minipage}{0.32\textwidth}
\centering
\begin{tabular}{|c|c|c|c|c|}
  \hline 
   & 1 & 2 & 2 & 2\\
  \hline
   &  &  & 1 & 2\\
  \hline
   &  &  & 1 & 2\\
  \hline
   &  &  &  & 1\\
  \hline
   &  &  &  & \\
  \hline
\end{tabular}
\end{minipage}
\caption{Anti-Monge matrix, reverse falling staircase anti-Monge matrix, and falling staircase anti-Monge matrix.}
\label{tab:monge}
\end{figure}

Let us now recall some basic facts regarding Monge matrices.
\begin{observation}
\label{obs:mon1}
Adding the same value to every element in a row (or a column) of an anti-Monge matrix results in an anti-Monge matrix.
\end{observation}

\begin{observation}
\label{obs:mon2}
To check if an array is anti-Monge it is sufficient to check if every contiguous $2 \times 2$ submatrix is anti-Monge.
\end{observation}

The following lemma follows from the well-known SMAWK algorithm~\cite{SMAWK}.
\begin{lemma}[Lemma 3.3 in Aggarwal et al.~\cite{AaK}]
\label{lem:ak}
All row maxima in a reverse falling staircase totally monotone matrix can be found in $O(n)$ time.
\end{lemma}

By transposing the matrix and observing that being anti-Monge implies being
totally monotone we obtain the following.
\begin{corollary}
\label{cor:am}
All column maxima in a falling staircase anti-Monge matrix can be found in $O(n)$ time.
\end{corollary}

\subparagraph{LIS-in-range queries.}
\label{ss:lis}
Let $S[1:n]$ be the input sequence. Define $M$ as an $(n+1) \times (n+1)$ matrix with $0$-indexed rows and columns, such that
$M[i, j]$ is the length of LIS in $S[i+1:j]$ for $i < j$ and $M[i, j] = j-i$ otherwise (the anti-Monge matrix in Figure~\ref{tab:monge} is
such a matrix for the sequence $(3, 4, 1, 2)$). As hinted by our example, this matrix turns out to have a rather special structure as observed by Tiskin~\cite{Tiskin}. We describe this structure in the following.

Let $S'$ be the sequence obtained by sorting $S$ (recall that $S$ consists of distinct elements), and observe that LIS of $S$ is the same as a longest
common subsequence (LCS, for short) of $S$ and $S'$. Thus, we can think that $M[i, j]$ is LCS of $S'$ and $S[i +1 : j]$. As such, the following result can be shown (see~\cite{Tiskin} 
and the references therein).

\begin{restatable}{lem}{lemamon}
\label{lem:amon}
$M$ is anti-Monge.
\end{restatable}

Our algorithm needs to access the elements of $M$. Since the matrix contains $(n+1)^{2}$ elements, it is too large to be explicitly
stored in memory. Fortunately, Tiskin also showed how to create in $O(n \log^2{n})$ time an $O(n)$-space implicit representation
of $M$ that allows us to obtain any of its elements in $O(\log{n})$ time~\cite{Tiskin}. Before we present the internals of this representation,
we need to introduce some additional definitions illustrated in Figure~\ref{example2}.

\begin{definition}
Let $A$ be any $n \times n$ matrix. Its distribution matrix $A^\Sigma$ is an $(n+1) \times (n+1)$ matrix defined by
$A^\Sigma[x, y] = \sum_{i\geq x, j<y}A[i, j]$, for every $1\leq x\leq n+1, 1\leq y\leq n+1$.
\end{definition}

\begin{definition}
A permutation matrix is a square matrix that has exactly one $1$ in every row and column, and the remaining elements are equal to $0$.
\end{definition}

\begin{figure}[htb]
\centering
\begin{minipage}{0.45\textwidth}
\centering
\begin{tabular}{|c|c|c|}
  \hline 
  0 & 1 & 0\\
  \hline
  1 & 0 & 0\\
  \hline
  0 & 0 & 1\\
  \hline
\end{tabular}
\end{minipage}
\hfill
%%%%%%%%%%%%%%%%%%%%%%%%%%%%%
\begin{minipage}{0.45\textwidth}
\centering
\begin{tabular}{|c|c|c|c|}
  \hline 
  0 & 1 & 2 & 3\\
  \hline
  0 & 1 & 1 & 2\\
  \hline
  0 & 0 & 0 & 1\\
  \hline
  0 & 0 & 0 & 0\\
  \hline
\end{tabular}
\end{minipage}
\caption{A permutation matrix $A$ and its distribution matrix $A^\Sigma$.}
 \label{example2}
\end{figure}

Now, we can provide the final ingredients of the construction. For two strings $w_1$ and $w_2$ of length $d$, Tiskin defines in~\cite{Tiskin} a $ (2d+1) \times (2d+1)$ matrix $L$ in the
following way. Let $w_2'$ be the string equal to $?^dw_2?^d$, whose positions are indexed from $-(d-1)$ to $2d$. The rows
of $L$ are indexed from $-d$ to $d$, while the columns of $L$ are indexed from $0$ to $2d$. The elements of $L$ are defined by
$L[i, j] = \text{LCS}(w_1, w_2'[i+1:j])$ if $j>i$, and $L[i, j] =j-i$ otherwise. In this definition, it is assumed that
$?$ matches any character. If $w_2$ is the input sequence $S$ and $w_1$ is $S'$ then, for $0 \leq i,j \leq n$
we have $L[i, j] = M[i+1, j+1]$. Tiskin proved (Theorem 4.10 in~\cite{Tiskin}) that there exists $2d \times 2d$
permutation matrix $P$ such that $L[i, j]$ = $j-i-P^\Sigma[i, j]$. Furthermore, he provided an $O(n\log^{2}n)$-time
algorithm that finds all the non-zero entries of $P$ (Algorithm 8.2 in~\cite{Tiskin}). Having all the non-zero entries
of $P$ we can apply a dominance counting structure of Chazelle~\cite{Ch} that can be constructed in $O(n\log n)$ time,
uses $O(n)$ space, and calculates $P^{\Sigma}[i,j]$ and hence also $M[i+1,j+1]$ in $O(\log n)$ time. Summarising,
in $O(n\log^{2}n)$ time we obtain a structure that returns any element of $M$ in $O(\log n)$ time. We
similarly obtain a matrix storing the length of LDS of every $S[i+1:j]$.

\subparagraph{Description of the algorithm.}
\label{ss:alg}

Let $S[1:n]$ be the input sequence. For every $1\leq x \leq n$, let~$\res[x]$ be the length of a longest $k$-rollercoaster
in $S[1:x]$, and $\inc[x]$ (respectively, $\dec[x]$) be the length of a longest $k$-rollercoaster in $S[1:x]$ with the last
run increasing (respectively, decreasing). Note that we do not require that these $k$-rollercoasters contain $S[x]$. 
Then, $\res[x] = \max\{ \dec[x], \inc[x]\}$, for $1\leq x \leq n$. Firstly, we~introduce~two~structural~lemmas.

\begin{restatable}{lem}{lemrollercoasterone}
\label{lem:rollercoaster1}
Let $A$ be a $k$-rollercoaster in $S[1:i]$ with the last run decreasing, and $r$ be an increasing subsequence in $S[i:n]$ such
that $\abs{r} \geq k$. Then there exists a $k$-rollercoaster in $S[1:n]$ of length at least $\abs{A} + \abs{r} - 1$ with the last
run increasing.
\end{restatable}

\begin{proof}
Let $A'$ be the sequence consisting of all elements from both $A$ and $r$. Recall that a sequence is a $k$-rollercoaster if every
run has length at least $k$. In order to show that $A'$ is a $k$-rollercoaster with last run increasing we need to consider three cases: the first element of $r$ is the last element of $A$, the first element of $r$ is
greater than the last element of $A$, and the first element of $r$ is less than the last element of $A$.

In the first case, all runs in $A'$ but the last are the same as in $A$, and the last run is equal to $r$. Since $A$ is a $k$-rollercoaster
and $\abs{r} \geq k$ we conclude that $A'$ is a $k$-rollercoaster. $A$ and $r$ have one common element, so $\abs{A'} =\abs{A} + \abs{r} - 1$.

In the second case, all runs in $A'$ but the last are also the same as in $A$, and the last run consists of the last element of $A$ and $r$. Again we conclude that $A'$ is a $k$-rollercoaster. Since $A$ and $r$ have no common elements, $\abs{A'} =\abs{A} + \abs{r}$.

In the third case, all runs in $A'$ but the last two are the same as in $A$. The second-to-last run in $A'$ consist of the
last run of $A$ and the first element of $r$, and the last run in $A'$ is $r$. Hence, $A'$ is a $k$-rollercoaster. Since $A$ and $r$
have no common elements, $\abs{A'} =\abs{A} + \abs{r}$.
\end{proof}

\begin{restatable}{lem}{lemrollercoastertwo}
\label{lem:rollercoaster2}
Consider a longest $k$-rollercoaster in $S[1:n]$ with the last run increasing (respectively, decreasing), and let $r$
be its last run with the first element $S[i]$. Then $r$ is a longest increasing (respectively, decreasing) subsequence in $S[i:n]$.   
\end{restatable}

\begin{proof}
By contradiction. Let $A$ be a longest $k$-rollercoaster from the statement of the lemma, and suppose that there exists a longer
increasing sequence $r'$ in $S[i:n]$. Let $A'$ be the prefix of $A$ ending at $S[i]$. Observe that $\abs{A'} = \abs{A} - \abs{r} + 1$.
Then by Lemma~\ref{lem:rollercoaster1} there exists a $k$-rollercoaster in $S$ of length at least
$\abs{A'}+\abs{r'}-1 = \abs{A} - \abs{r} + \abs{r'} > \abs{A}$.
\end{proof}

The above lemmas allow us to obtain the formula for calculating the arrays $\inc$ and $\dec$.
Recall that $M[i, j]$ is the length of LIS in $S[i+1:j]$. Let $M'$ be the matrix obtained from $M$ by replacing
all elements less than $k$ by $-\infty$, and let $Z(j,j')$ be the set of indices $j \leq i \leq j'$
such that length of LIS in $S[i:j']$ is at least $k$ (or, in other words, $M'[i-1, j'] \neq -\infty$).

\begin{restatable}{proposition}{propinc}
\label{prop:inc}
For every $1\leq x\leq n$, the following holds:
\begin{align*}
\inc'[x] = \max\{ \dec[i] + M'[i-1, x]-1 : i \in Z(1,x) \}, 
\inc[x] = \max\{\inc'[x], M'[0, x]\}.
\end{align*}
If $Z(1, x)$ is empty then we set $\inc'[x] = 0$.
\end{restatable}

\begin{proof}
By Lemma~\ref{lem:rollercoaster1} we obtain that for every $i \in Z(1, x)$ there exists a $k$-rollercoaster in $S[1:x]$ with the
last run increasing of length at least $\dec[i] + M'[i-1, x]-1$. We conclude that $\inc'[x]$ is less or equal to the length of a
longest $k$-rollercoaster with the last run increasing in $S[1:x]$. Observe that $M'[0, x]$ corresponds to an increasing run
of length at least $k$ or is equal to $-\infty$. We obtain that $\inc[x]$ is less or equal than the length of a longest $k$-rollercoaster
with the last run increasing in $S[1:x]$.

For the converse, consider a $k$-rollercoaster $A$ with the last run increasing in $S[1:x]$. If $A$ consists of just a single
run then its length is $M'[0,x]$. Otherwise, let $S[i]$ be the first element in the last run of $A$. Then by Lemma~\ref{lem:rollercoaster2} the
length of the last run is equal to $M'[i-1, x]$ and the length of $A$ is $\dec[i] + M'[i-1,x] - 1$. Overall, the length of $A$
is at most $\inc[x]$.
\end{proof}

Proposition~\ref{prop:inc} cannot be applied directly if we aim to achieve the announced $O(n\log^{2}n)$ time complexity,
and we need to introduce some auxiliary definitions.
For every $1\leq d \leq x$ we define $\inc_d[x]$ as follows:
$$
\inc_d'[x] = \max\{\dec[i] + M'[i-1, x]-1 : i \in Z(1, d-1)  \},
\inc_d[x] = \max\{\inc_d'[x], M'[0, x]\} .
$$
If $Z(1,d-1)$ is empty then we set $\inc_d'[x] = 0$.
In other words, $\inc_d[x]$ is equal to the length of a longest $k$-rollercoaster
in $S[1 : x]$ with the last run increasing and starting at an element $S[i]$ with $i<d$ or LIS of $S[1:n]$ of length at least $k$.
Thus, $\inc_{1}[x]$ is equal to either $0$ or
the length of a LIS in $S[1 : x]$. We similarly define $\dec_{d}[x]$.
\begin{observation}
\label{obs:inc}
For every $j > i-k+1$, $\inc_{j}[i] = \inc[i]$.
\end{observation}

We describe a function \textsc{Compute} that receives a contiguous subsequence $S[i:j]$ together with the previously
calculated arrays $\inc_{i}[i:j]$ and $\dec_{i}[i:j]$, and returns the arrays $\inc[i:j]$ and $\dec[i:j]$.
To calculate the length of a longest $k$-rollercoaster in $S[1:n]$ we invoke the function with the whole $S[1:n]$ and the
arrays $\inc_1[1:n]$, $\dec_1[1:n]$ as arguments, and return the maximum over the two resulting arrays. 
Note that $\inc_1[1:n]$ and $\dec_1[1:n]$ can be calculated in $O(n\log n)$ time using Algorithm \ref{alg:lis}. 

Let $m = \ceil{\frac{i+j}{2}}$. The main idea of \textsc{Compute} is to call the function recursively for the left half to calculate $\inc[i:m-1]$
and $\dec[i:m-1]$. The next step is to calculate $\inc_m[m:j]$ and $\dec_m[m:j]$ using tools from the previous paragraphs
(as described below). Finally, we recursively calculate $\inc[m:j]$ and $\dec[m:j]$. Concatenating the results from both
recursive calls gives us the desired result. This is summarised in Algorithm~\ref{alg:log2}.

\begin{algorithm}
\caption{Computing the length of a longest $k$-rollercoaster}
\label{alg:log2}
\begin{algorithmic}[1]
\Procedure{Compute}{$k$, $S[i : j]$, $\inc_i[i : j]$, $\dec_i[i : j]$}
\If {$j-i+2 \leq k$}
\State $\{ \inc[i:j], \dec[i:j] \} \gets \{ \inc_i[i:j], \dec_i[i:j] \}$
\State \textbf{return} $\{\inc[i:j], \dec[i:j]\}$
\EndIf
\State $m \gets \ceil{\frac{i+j}{2}}$
\State $\{\inc[i:m-1], \dec[i:m-1]\} \gets \textsc{Compute}(k, S[i : m-1], \inc_i [i : m-1], \dec_i [i : m-1])$
\State Compute $\inc_m[m:j]$ and $\dec_m[m:j]$
\State $\{\inc[m:j], \dec[m:j]\} \gets \textsc{Compute}(k, S[m : j], \inc_m[m : j], \dec_m[m : j])$
\State \textbf{return} $\{\inc[i:j], \dec[i:j]\}$
\EndProcedure
\end{algorithmic}
\end{algorithm}

\subparagraph{Computing $\inc_m[m:j]$ and $\dec_m[m:j]$.}
We only describe how to calculate $\inc_{m}[m:j]$, as $\dec_{m}[m:j]$ can be computed by a similar approach.
Recall the previously introduced matrix $M'$, obtained by replacing values less than $k$ by $-\infty$ in $M$.
Let $A_{\inc}$ be the $ (m-i) \times (j+1-m)$ matrix with rows indexed from $i$ to $m-1$ and columns indexed from $m$ to $j$ satisfying:
\[ A_{\inc}[x, y] = 
	\begin{cases}
	\dec[x] + M'[x-1, y] - 1 &\text{when } M'[x-1, y] \neq -\infty,\\
	\text{blank} &\text{otherwise.}\\
	\end{cases}
\]
Since we are able to retrieve any element of $M'$ in $O(\log n)$ time using LIS-in-range queries,
and the value of $\dec[x]$, for every $i \leq x \leq m-1$, is already available, each element of $A_{\inc}$
can be calculated in  $O(\log n)$ time. Furthermore, we have the following property.

\begin{proposition}
$A$ is a falling staircase anti-Monge matrix.
\end{proposition}

\begin{proof}
By Lemma \ref{lem:amon} $M$ is an anti-Monge matrix. By Observation~\ref{obs:mon1} this is still the case if we add the same
value to all elements in the same row.

To prove that $A$ is a falling staircase matrix consider a non-blank element
$A[i,j]$. Then $M[i,j] \geq k$. But this implies $M[i-1,j] \geq k$ and $M[i,j+1] \geq k$ (as long as $i > 1$ and $j < n$), so all elements above and to the right are
also non-blank as required.
\end{proof}

\begin{proposition}
\label{prop:log2}
For every $m \leq \ell \leq j$, $\inc_m[\ell]$ is equal to either $\inc_i[\ell]$ or the maximum in the $\ell^{\text{th}}$ column of $A$.
\end{proposition}
\begin{proof}
For every $m\leq \ell \leq j$, $\inc_{m}[\ell]$ is equal to either $\inc_{i}[\ell]$ or
$\max\{  \dec[j] + M'[j-1, \ell]-1 : j \in Z(i,m-1) \}$. However, the latter is exactly the maximum in the $\ell^{\text{th}}$ column of $A$.
\end{proof}

\begin{lemma}
\label{lem:zlozonosc}
We can compute $\inc_m[m:j]$ and $\dec_m[m:j]$ in $O((j-i+1)\log{n})$ time.
\end{lemma}

\begin{proof}
By Proposition~\ref{prop:log2} computing $\inc_m[m:j]$ reduces to finding all the column maxima in $A$.
Since $A$ is a falling staircase anti-Monge matrix, we can use the algorithm from Corollary~\ref{cor:am}.
Access to any element of $A$ requires $O(\log n)$ time, so in total we obtain $O((j-i+1)\log n)$ time complexity.
\end{proof}

We can now state with the main result of this section.
\begin{theorem}
For every sequence $S[1:n]$ and $k\geq 3$, the length of a longest $k$-rollercoaster in $S$ can be found in $O(n \log^2{n})$ time.
\end{theorem}

\begin{proof}
The algorithm needs $O(n \log^2{n})$ preprocessing time to construct the LIS-in-range (and LDS-in-range) structure.
We compute $\inc_1[1:n]$ and $\dec_1[1:n]$ in $O(n \log n)$ time using Algorithm \ref{alg:lis}.
Then, we call the recursive function $\textsc{Compute}$. 
By Lemma~\ref{lem:zlozonosc} a call of the function on $S[i:j]$ takes $O((j-i+1)\log n)$ time,
so its running time is described by the recurrence $T(n) = 2T(n/2) + O(n\log n)$ that solves to $O(n\log^{2}n)$.
Thus, the overall time complexity is $O(n \log^2{n})$.
\end{proof}

\section{Constructing a Longest $k$-Rollercoaster}
\label{sec:recover}
In this section we briefly discuss how to construct a longest rollercoaster for both algorithms.

\subparagraph{For the $O(nk^2)$ algorithm.} In the respective algorithm, for each $2\leq i\leq k$, and for each element $S[j]$, we compute the predecessor of $S[j]$ on a longest (not necessarily contiguous) subsequence of $S$ ending with $S[j]$ and with every run of length at least $k$, except for the last run, which has only $i$ elements if $i < k$ and at least $k$ elements if $i = k$. If, together with this predecessor, we store also the length of the last run in the respective subsequence of $S$, we can trace the whole sequence back. Indeed, the predecessor gives us the information what element should we list before $S[j]$ in the subsequence. The length of the run gives us information on the length of the run ending with the predecessor of $S[j]$, so we know where we should look in our data structures for the predecessor of $S[j]$. For some $i$ and $j$, tracing back a longest (not necessarily contiguous) subsequence of $S$ ending with $S[j]$ and with every run of length at least $k$, except for the last run, which has only $i$ elements if $i < k$ and at least $k$ elements if $i = k$, takes, clearly, $O(n)$ time, provided that we have the information described above.

In the end, we will only need to trace back a longest (not necessarily contiguous) subsequence of $S$ ending with some element $S[j]$ and with every run of length at least $k$. Given that we also compute the length of a longest (not necessarily contiguous) subsequence of $S$ ending with each $S[j]$ and with every run of length at least $k$, we can select in $O(n)$ time the ending element of the subsequence we need to trace back.

In conclusion, once the $O(nk^2)$ time algorithm for computing the length of a longest $k$-rollercoaster is executed, we can actually compute the respective $k$-rollercoaster in $O(n)$ additional time. 

\subparagraph{For the $O(n\log^2 n)$ algorithm.}
In order to retrieve the elements of a longest $k$-rollercoaster we need to extend our algorithm to maintain global arrays
$Pred_{\inc}[1, \ldots, n]$ and $Pred_{\dec}[1, \ldots, n]$. Elements of these arrays are computed during the calculations of
$\inc_m[m:j]$ and $\dec_m[m:j]$ as follows. Initially they are equal to $-1$. After execution of the algorithm we demand that
$Pred_{\inc}$ satisfies the following: for every $1\leq i \leq n$ we have that $\inc[i] = \dec[ Pred_{\inc}[i] ] + M'[Pred_{\inc}[i]-1,i] -1$ if $Pred_{\inc}[i] \neq -1$ and $\inc[i] = \max \{ 0, M'[0, i] \}$ otherwise, and similarly for $Pred_{\dec}$. It is straightforward to augment the algorithm from Corollary~\ref{cor:am}
to obtain such information. 
 
We retrieve the elements of a longest $k$-rollercoaster from the last one to the first one. Recall that a longest
$k$-rollercoaster has the length equal to $\max\{\inc[n], \dec[n]\}$. We focus on how to obtain a longest $k$-rollercoaster
$R$ of length $\inc[n]$ with last run increasing (so, assume, w.l.o.g., that $\inc[n]>\dec[n]$); the procedure is similar for
$\dec[n]$ and the last run decreasing. 

Observe that if $\inc[n]$ is equal to the length of  LIS in the input sequence, we can obtain the elements of $R$ by Algorithm \ref{alg:lis} in $O(n \log{n})$ time. Otherwise, there exists $i<n$ such that $\inc[n] = \dec[i] + M[i-1, n] - 1$. The value $i$ is stored in $Pred_{\inc}[n]$. In this case, we construct $R$ by finding recursively a longest $k$-rollercoaster associated with $\dec[i]$ and concatenating it with LIS in $S[i:n]$. This holds because, by Lemma~\ref{lem:rollercoaster2} the last run of $R$ is a LIS in $S[i:n]$. Obtaining LIS in $S[i:n]$ can be done in $O((n-i) \log n)$ time. 

Thus, in general, we will need to compute a series of LISs and LDSs on the ranges $S[n_{i-1}:n_{i}]$, for $1\leq i\leq m$, where $m$ is the number of runs in a longest $k$-rollercoaster, $n_m=n$ and $n_0=1$. Moreover, $n_{i-1}=Pred_{\inc}[n_i]$, if the $i^\text{th}$ run of the rollercoaster is increasing and $n_{i-1}=Pred_{\dec}[n_i]$, if the $i^\text{th}$ run of the rollercoaster is decreasing. Obtaining the LIS in $S[n_{i-1}:n_{i}]$ can be done in $O((n_i-n_{i-1}+1) \log{n})$ time. 

Adding up the time needed to compute LIS or LDS for each of these ranges we get $O(n\log{n})$ total time needed to obtain elements of a longest $k$-rollercoaster.

\section{Lower Bound}

In the final section of our paper, we prove that any comparison-based algorithm computing the length of a longest $k$-rollercoaster in
a permutation $S$ of $\{1,\ldots,n\}$, for $4 \leq k \leq \frac{n}{3}$, performs at least $\Omega(n \log{k})$ comparisons.
Let $T$ be a binary comparison tree associated with an algorithm that computes the result. The number of comparisons made in the algorithm is equal
to the height of $T$, and this is a lower bound on the execution time of the algorithm. 

Let $A$ be a partial ordering associated with a path from the root to some leaf of $T$. Since the algorithm cannot distinguish between
permutations following the same path, every permutation consistent with $A$ has to give the same result. Our approach is to first
identify a set $U$ of permutations of $\{1,\ldots,n\}$ such that $\log{\abs{U}} = \Theta(n \log {k})$, and any ordering associated
with a leaf of $T$ can be consistent with at most one permutation from $U$. Hence, the number of leaves in $T$ is at least $\abs{U}$.
Since the height of a binary tree is at least logarithm of the number of leaves,
this will show that the height of $T$, and hence also the number of comparison performed by the algorithm, is at least $\Omega(\log{\abs{U}}) = \Omega(n \log{k})$.

We first recall the set $\Gamma$ of $\ell^{n-2\ell}$ permutations of $\{1,\ldots,n\}$ proposed by
Fredman in~\cite{Fredman75}, where $\ell$ is a parameter. These permutations are essentially different inputs $S$ for an
algorithm computing the length of LIS, each leading to a different leaf in the comparison tree.

So, essentially, we want to construct input sequences $(x_1,\ldots,x_n)$, with their elements $x_1,\ldots,x_n$ chosen so that
certain linear orderings of the $x_i$s are induced. To create a permutation from $\Gamma$ we partition $(x_1,\ldots,x_n)$ into
$\ell$ subsequences $P_1, P_2, \ldots, P_{\ell}$. 
To simplify the exposure, let $\ell_{\text{prefix}}$ of a sequence be its prefix of length $\ell$, while the
$\ell_{\text{suffix}}$ is its suffix of length $\ell$; the remaining $n - 2\ell$ elements are called $\ell_{\text{middle}}$ of the sequence. We partition $(x_1,\ldots,x_n)$ in the following way: the $i^\text{th}$ element of $\ell_{\text{prefix}}$
(that is, $x_i$) and the $i^\text{th}$ element of $\ell_{\text{suffix}}$ ($x_{n-\ell+i}$) belong to $P_i$. Each element from $\ell_{\text{middle}}$ of the
sequence belongs to an arbitrary chosen part $P_j$. This gives us $\ell^{n-2\ell}$ different partitions. For a partition
$P_1,\ldots,P_\ell$, we assign values from $\{1,\ldots,n\}$ to the input sequence in such a way, that the elements of each part
$P_i$ form a decreasing sequence and, for $1\ \leq i \leq \ell-1$, each element of $P_i$ is less than any element of $P_{i+1}$
(see Figure~\ref{fig:perm}). So, each such possible assignment gives us a permutation from $\Gamma$. LIS of any permutation from
$\Gamma$ is of length $\ell$ because it contains one element from each $P_i$. LDS of any permutation of $\Gamma$ is no
longer than $n-2\ell+2$ because it contains at most one element from $\ell_{\text{prefix}}$ and at most one from $\ell_{\text{suffix}}$.

\begin{proposition}
\label{pro:split}
Each permutation from $\Gamma$ can be split into $\ell$ descending subsequences in only one way. For two different permutations
from $\Gamma$ these ways of splitting are different. 
\end{proposition}

\begin{proof}
Let $P$ be a permutation from $\Gamma$ and $P_1, \ldots, P_{\ell}$ be its corresponding partition as described above.
Observe that elements of $\ell_{prefix}$ (respectively, $\ell_{suffix}$)
of $P$ form an increasing subsequence, so no two of them can be in the same decreasing subsequence.  
Now let $D_1,\ldots,D_{\ell}$ be a partition of $P$ into $\ell$ decreasing subsequences, such that $D_i$
contains the $i^\text{th}$ element from $\ell_{prefix}$.
Since elements of $\ell_{suffix}$ form an increasing subsequence, each $D_i$ has to
contain exactly one of them.
Because only the first element in $\ell_{suffix}$ is smaller than the first element in the
$\ell_{prefix}$, $D_{1}$ actually has to contain the first element of $\ell_{suffix}$.
Repeating this reasoning, we obtain that $D_{i}$ contains the $i^\text{th}$ element
from $\ell_{prefix}$ and also the $i^\text{th}$ element from $\ell_{suffix}$.
Then, we obtain that $D_{1}$ is actually equal to $P_{1}$, and by repeating
this reasoning, that $D_{i}=P_{i}$ for all $i=1,\ldots,\ell$.
\end{proof}

\begin{figure}[htb]
\centering
\includegraphics[width=0.7\textwidth]{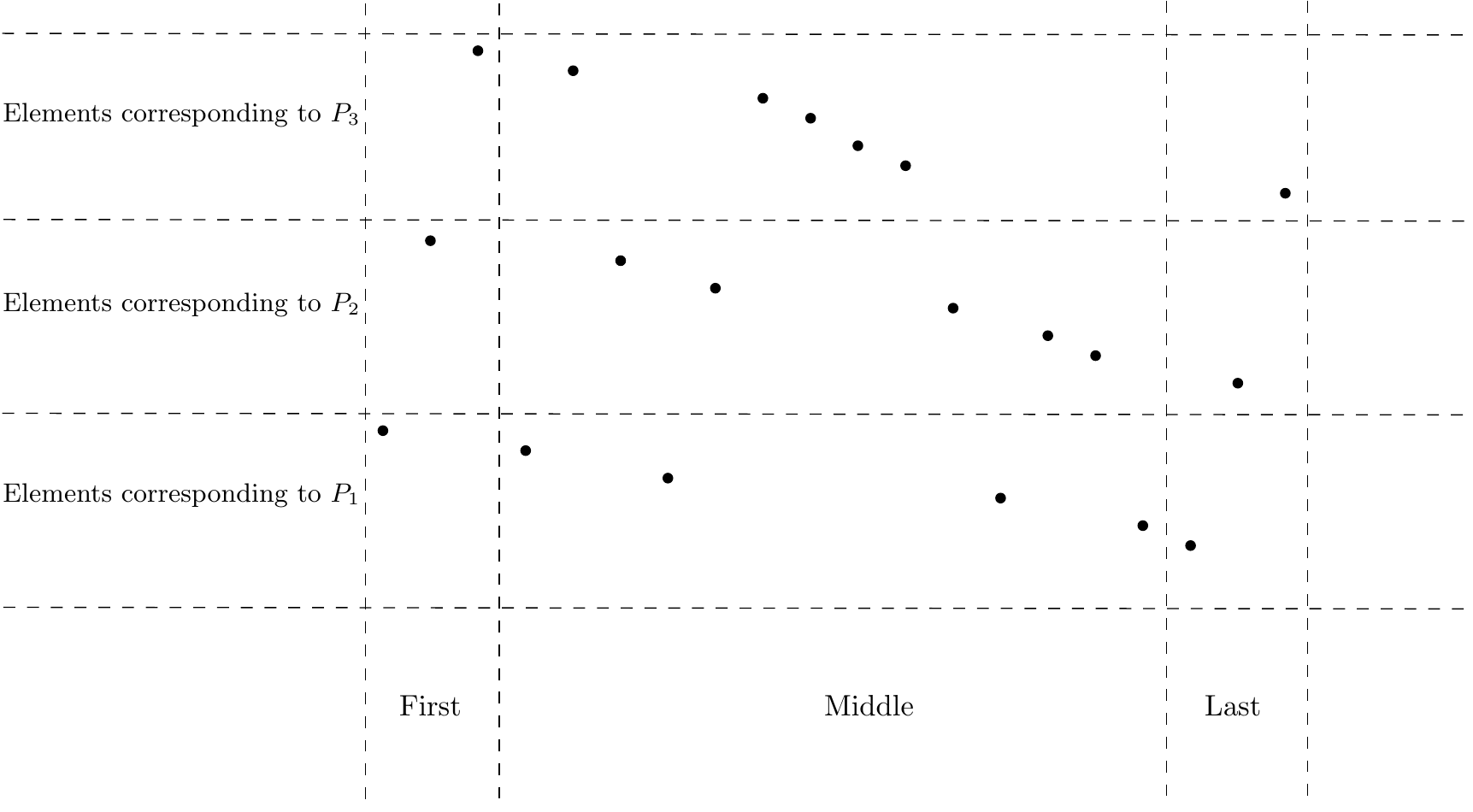}
\caption{Example permutation $P \in \Gamma$ for $\ell = 3$ in a~plane. In this figure, we have $P = (6,13,20,5,19,12,4,11,18,17,16,15,10,3,9,8,2,1,7,14)$.}
\label{fig:perm}
\end{figure}

We now consider the algorithm computing the length of a longest $k$-rollercoaster. Using the permutations from $\Gamma$ we create a set $U$ of $k^{n \frac{k-3}{3k-3}}$ permutations of $\{1,\ldots,n\}$, again with the same principle behind: they should be input sequences which lead to different paths in the comparison tree associated to an algorithm computing the length of a longest $k$-rollercoaster. Observe that $\log{(k^{n \frac{k-3}{3k-3}})} = \Theta(n \log{k})$, so this would imply the desired lower bound of $\Theta(n \log{k})$ on the number of comparisons done by an algorithm to compute the length of a longest $k$-rollercoaster.

A permutation from $U$ is obtained as follows. Suppose that $(3k-3)$ divides $n$. Split the sequence $(x_1,\ldots,x_n)$ into $\frac{n}{3k-3}$
{\em blocks} (contiguous subsequences) of size $3k-3$. We will assign to the elements of the $i^\text{th}$ contiguous block
$(x_{i(3k-3)+1},\ldots,x_{(i+1)(3k-3)})$ distinct values from the set  $\{{i(3k-3)+1},\ldots,{(i+1)(3k-3)}\}$, as follows. In every block,
use one of the permutations from $\Gamma$ (with the parameter $\ell$ set to $k$) to values to the elements $x_{i(3k-3)+1},\ldots,x_{(i+1)(3k-3)}$
of that block, and then assign values to those elements according to
that ordering. In this way, we can create $|\Gamma|^{ \frac{n}{3k-3} } = (k^{k-3})^{\frac{n}{3k-3}}$ permutations of $\{1,\ldots,n\}$.
Observe that in every block the length of a longest decreasing subsequence is less than $k$. Since every block consists of strictly greater
values than the previous ones, a longest decreasing subsequence of every permutation from $U$ is less than $k$. A longest increasing
subsequence of every element of $\Gamma$ is equal to $k$, so a longest $k$-rollercoaster for every element of $U$ is equal to
$\frac{kn}{3k-3}$ and consists only of longest increasing subsequences corresponding to all the blocks glued one after the other. We can now show a result similar to
Proposition~\ref{pro:split}.

\begin{proposition}
\label{pro:partition}
Each permutation from $U$ can be split into $\frac{kn}{3k-3}$ descending subsequence in only one way. For two different permutations from $U$ these ways of splitting are different. 
\end{proposition}
\begin{proof}
Let $S$ be a permutation from $U$. Recall that we can partition $S$ into $\frac{n}{3k-3}$ contiguous blocks of length $3k-3$. All values in a block are strictly greater
than the values in all previous blocks,
so in a decreasing subsequence of $S$ we can have only elements
from one block. Since every block corresponds to a permutation from $\Gamma$, by Proposition \ref{pro:split} it can be split into exactly
$k$ decreasing subsequences in only one way. For each two different permutations of $U$, there exists at least one block (i.e., permutation from
$\Gamma$) that differentiates them. By Proposition~\ref{pro:split}, this block is split in a different way
than all the other blocks of $\Gamma$, so the conclusion follows: each particular permutation from $U$ will also be split in a different
way than all other permutations of $U$.
\end{proof}

Having constructed the set $U$, we can proceed with the lower bound. 
Let $A$ be a partial ordering associated with a path to some leaf of $T$ (the comparison tree associated to the algorithm computing the length of a longest $k$-rollercoaster). Since the algorithm cannot distinguish between permutations
following the same path, every permutation consistent with $A$ has to give the same result. We recall the following lemma.

\begin{lemma}[Lemma 3.6 in \cite{Fredman75}]
\label{lem:fredman}
Let $\leq$ be a partial ordering defined on $S$. The maximum length of LIS in $S$ associated with any linear embedding of this ordering, is equal to the minimum number of decreasing subsequences relative to $\leq$ into which $S$ can be partitioned.
\end{lemma}

Now we can prove the following.

\begin{lemma}
Let $A$ be partial ordering associated with the path from the root to a leaf of $T$. Only one permutation from $U$ can be consistent with $A$.
\end{lemma}

\begin{proof}
Consider $S\in U$ that is consistent with $A$, and let
$D=\frac{kn}{3k-3}$ be the length of its LIS.
Now let $m$ be the minimum number of decreasing subsequences
relative to the results of the comparisons made on the path $A$ into which
$S$ can be partitioned. 
If $m<d$ then $S$ is consistent with $A$, so we can partition $S$ into the same decreasing subsequences, but $S$ cannot be divided into less than than
$d$ decreasing subsequences, a contradiction.
If $m>d$ then by Lemma \ref{lem:fredman} there exists a permutation $S'$ consistent with $A$ with the length of LIS greater than $d$. $S'$ follows the same path as $S$ in the comparison tree, but has a longer $k$-rollercoaster (consisting only of LIS of $S'$) than $S$, a contradiction. 
Thus, $m=d$ for any such $S$.

Consider two $S_{1},S_{2}\in U$ consistent with $A$.
By Proposition \ref{pro:partition}, the only partition of $S_1$ into $d$ decreasing sequences is different from the only such partition of $S_2$ (into $d$ decreasing sequences), so $A$ can be consistent with only one permutation,
a contradiction.
\end{proof}

Thus, each permutation from $U$ corresponds to a distinct leaf of $T$, making the depth of $T$
at least $\log\abs{U} = \Theta(n \log{k})$ as required and proving the following theorem.

\begin{theorem}
\label{thm:lb}
For every $k$ satisfying $4 \leq k \leq \frac{n}{3}$, any comparison-based algorithm that computes the length of a longest $k$-rollercoaster
in a permutation of $\{1,\ldots,n\}$ performs at least $\Omega(n \log{k})$ comparisons.
\end{theorem}

\bibliography{biblio}

\end{document}